\documentclass[a4paper,11pt]{article} 
\usepackage[normalem]{ulem}
\usepackage{authblk}
\usepackage[ruled,vlined,linesnumbered]{algorithm2e}
\usepackage[noend]{algpseudocode}
\usepackage{amsmath}
\usepackage{multirow}
\usepackage{amsfonts}
\usepackage{amssymb} 
\usepackage{amsthm}
\usepackage{complexity}
\usepackage{fullpage}
\usepackage{pstricks} 
\usepackage{graphicx,url}
\usepackage{hyperref}
\usepackage{lineno}
\usepackage{setspace}
\usepackage{tikz}
\usepackage{pstricks}
\usepackage{enumerate}
\usepackage{tabularx}
\newtheorem{theorem}{Theorem}[section]

\newtheorem{proposition}[theorem]{Proposition}
\newtheorem{corollary}[theorem]{Corollary}
\newtheorem{claim}{Claim}

\newenvironment{reduction}{\bigskip\noindent {\bf Reduction A.}}{\qed}

\newcounter{CounterTodo}
\usepackage{xcolor}

\newenvironment{proofClaim}[1]{\noindent {\em Proof of Claim~$\ref{#1}$.}}{\hfill$\blacksquare$\bigskip}

\usepackage{arydshln}
\usepackage{hyperref}
\hypersetup{
    colorlinks=true,
    linkcolor=blue,
    filecolor=blue,      
    urlcolor=blue,
}

\begin{document}

\onehalfspace

\title{Expanded-clique graphs and the domination problem}

\author[1]{Mitre~C.~Dourado\thanks{mitre@ic.ufrj.br. Partially supported by CNPq grant numbers 305404/2020-2 and 403601/2023-1, and FAPERJ grant number SEI-260003/015314/2021}}
\author[2]{Rodolfo~A.~Oliveira\thanks{rodolfooliveira@id.uff.br}}
\author[1]{Vitor~Ponciano\thanks{vtponciano@ufrj.br}}
\author[1]{\\ Alessandra~B.~Queiróz\thanks{alessandrabarbosa@ufrj.br}}
\author[1]{Rômulo~L.~O.~Silva\thanks{romulolo@ufpa.br}}

\affil[1]{Instituto de Computação, Universidade Federal do Rio de Janeiro, \linebreak Rio de Janeiro, Brazil} 
\affil[2]{Instituto do Noroeste Fluminense, Universidade Federal Fluminense, Brazil}

\maketitle

\begin{abstract}
Given a graph $G$ such that each vertex $v_i$ has a value $f(v_i)$, the expanded-clique graph $H$ is the graph where each vertex $v_i$ of $G$ becomes a clique $V_i$ of size $f(v_i)$ and for each edge $v_iv_j \in E(G)$, there is a vertex of $V_i$ adjacent to an exclusive vertex of $V_j$.
In this work, among the results, we present two characterizations of the expanded-clique graphs, one of them leads to a linear-time recognition algorithm.
Regarding the domination number, 
we show that this problem is \NP-complete for planar bipartite $3$-expanded-clique graphs and for cubic line graphs of bipartite graphs.
\end{abstract}

\section{Introduction} \label{sec:int}

We consider finite, simple and undirected graphs.
The degree, the open and the closed neighborhoods of a vertex $v$ are denoted by $d(v), N(v)$ and $N[v]$, respectively.
In this text, whenever we refer to a graph by $G$, we will denote its vertex set by $\{v_1, \ldots, v_n\}$.
Given a function $f : V(G) \rightarrow \mathbb{N}$ such that $f(v_i) \ge d(v_i)$ for every $v_i \in V(G)$, the {\em expanded-clique graph} $H$ of $(G,f)$ is defined as follows: for every $v_i \in V(G)$, there is a set $V_i \subseteq V(H)$ with $f(v_i)$ vertices forming a clique; and for every $v_iv_j \in E(G)$, there are 
$v_{i,j} \in V_i$ and $v_{j,i} \in V_j$ such that $v_{i,j}v_{j,i}$ in $E(H)$.
If $f(v_i) - d_G(v_i) = k_i > 0$, then $V_i$ has $k_i$ simplicial vertices, which are denoted by $v'_{i,1}, \ldots, v'_{i,k_i}$.
In this case, $G$ is the {\em root of $H$} under the $f$-expanded-clique operation. 
The set $V_i$ will be refered to as the {\em expanded clique $($associated with $v_i)$}.
Note that for every $v \in V_i$, $d(v) \in \{|V_i|-1, |V_i|\}$.
We say that $H$ is an {\em expanded-clique graph} if $H$ is the expanded-clique graph for some graph $G$ and function $f$.
If $f(v_i) = k$ for some $k \in \mathbb{N}$, then we can say that $H$ is a {\em $k$-expanded-clique} graph. 

In this work, we are interested in the complexity aspects of the recognition problem of the expanded-clique graphs and of the domination problem on this class. We remark that the well-studied inflated graphs~\cite{FAVARON1998,HENNING2012164,KANG2004485} form a subclass of the clique-expanded graphs, since a graph $H$ is {\em inflated} if it is the expanded-clique graph of a pair $(G,f)$ satisfying $f(v_i) = d(v_i)$ for every $v_i \in V(G)$.

Recall that $D \subseteq V(G)$ is a \textit{dominating set} of a graph $G$ if every vertex of $V(G) \setminus D$ has a neighbor in $D$.
The domination problem is a classic problem in graphs having many relevant applications~\cite{book1998}. It can be stated as:

\bigskip
\noindent {\sc Dominating set} \\
\begin{tabular}{lp{14cm}}
	Input: & A graph $G$ and a positive integer $\ell$. \\
	Question: & Is there a dominating set $S\subseteq V(G)$ so that $|S|\leq \ell$?
\end{tabular}
\bigskip

The text is organized as follows.
In Section~\ref{ClasseAndRec},
we begin by showing that the subdivided-line graphs~\cite{Hasunuma} and the Sierpi\'nski graphs~\cite{klavzar2002} are proper subclasses of expanded-clique graphs, and that this class is a proper subclass of the line graphs of bipartite graphs~\cite{Harary74}.
Next, we present two characterizations of the expanded-clique graphs, one of them leading to a linear-time recognition algorithm.
Section~\ref{SecDomination} deals with the domination problem.
We show that this problem is $\NP$-complete for 
planar bipartite $3$-expanded-clique graphs and for
cubic line graphs of bipartite graphs.
We also show that given an expanded-clique $H$, the domination number of $H$ plus the 2-independence number of the root $G$ of $H$ is equal to $|V(G)|$.
As a consequence of this result, we derive lower and upper bounds for the dominating number of expanded-clique graphs, which lead to the fact that a dominating set of $H$ can be easily found within the ratio $1+\frac{1}{\Delta(G)}$ of the minimum.

We conclude this section by presenting useful notation.
Consider $G$ a graph. The minimum and maximum degrees of $G$ are denoted by $\delta(G)$ and $\Delta(G)$, respectively. We say that $G$ is a \textit{cubic graph} if $G$ is $3$-regular.
We write $K_n$ for the complete graph with $n$ vertices.
The subgraph of $G$ \textit{induced} by $V'$ is denoted by $G[V']$.
A vertex $v\in V$ is a \textit{simplicial vertex} if $N(v)$ induces a clique.
For the case of $V'\subseteq V$, denote the \textit{closed neighborhood of $V'$ in $G$} as $N[V']=\{v\in N[v']:\mbox{ for all } v'\in V'\}$, and \textit{open neighborhood of $V'$ in $G$} as $N(V') = N[V']\setminus V'$.
We write $K_{p,q}$ for the complete bipartite graph where independent sets have respectively the sizes $p$ and $q$. 
A \textit{claw} is a $K_{1,3}$ graph.
A \textit{diamond} is an induced cycle $C_4$ plus one chord.
An \textit{odd-hole} is any induced cycle $C_q$ where $q$ is an odd number greater than 4.
The {\em butterfly graph}, also called hourglass graph, is the graph depicted in Figure~\ref{CDHB}. 

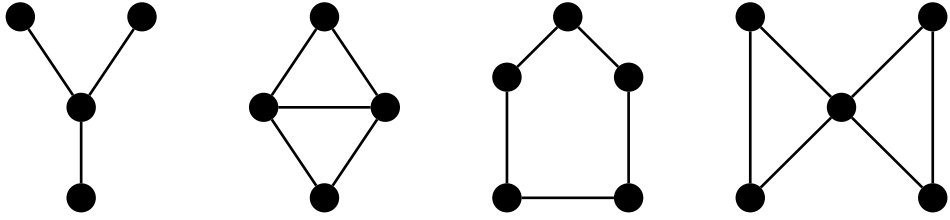
\begin{figure}[!htb]
    
    \vspace{10pt}
    \centering
    
    \begin{tikzpicture}[scale=0.4]
        \pgfsetlinewidth{1pt}
        
        \tikzset{
            vertex/.style={circle,  draw, minimum size=10pt, inner sep=0pt}}
 
        \begin{scope}
               
            \node [vertex,fill=black] (u1) at (0,3){};
            \node [vertex,fill=black] (u2) at (4,3){};
            \node [vertex,fill=black] (u3) at (2,0){};
            \node [vertex,fill=black] (u4) at (2,-3){};

            \draw[-] (u1) to (u3);
            \draw[-] (u2) to (u3);     
            \draw[-] (u3) to (u4);

      \end{scope}

  \begin{scope}[shift={(8,0)}]
               
            \node [vertex,fill=black] (v1) at (2,3){};
            \node [vertex,fill=black] (v2) at (0,0){};
            \node [vertex,fill=black] (v3) at (4,0){};
            \node [vertex,fill=black] (v4) at (2,-3){};

            \draw[-] (v1) to (v2);
            \draw[-] (v1) to (v3);     
            \draw[-] (v2) to (v3);
            \draw[-] (v2) to (v4);
            \draw[-] (v3) to (v4);

      \end{scope}

  \begin{scope}[shift={(16,0)}]
               
            \node [vertex,fill=black] (v1) at (2,3){};
            \node [vertex,fill=black] (v2) at (0,1){};
            \node [vertex,fill=black] (v3) at (0,-3){};
            \node [vertex,fill=black] (v4) at (4,-3){};
            \node [vertex,fill=black] (v5) at (4,1){};    
 
            \draw[-] (v1) to (v2);
            \draw[-] (v2) to (v3);     
            \draw[-] (v3) to (v4);
            \draw[-] (v4) to (v5);
            \draw[-] (v1) to (v5);
            
      \end{scope}

   \begin{scope}[shift={(24,0)}]
               
            \node [vertex,fill=black] (v1) at (0,3){};
            \node [vertex,fill=black] (v2) at (0,-3){};
            \node [vertex,fill=black] (v3) at (3,0){};
            \node [vertex,fill=black] (v4) at (6,3){};
            \node [vertex,fill=black] (v5) at (6,-3){};    
 
            \draw[-] (v1) to (v2);
            \draw[-] (v1) to (v3);     
            \draw[-] (v2) to (v3);
            \draw[-] (v3) to (v4);
            \draw[-] (v3) to (v5);
            \draw[-] (v4) to (v5);
            
      \end{scope}
    \end{tikzpicture}
    \label{CDHB}
\caption{Claw, Diamond, Odd-hole and Butterfly graph}
\end{figure}

\section{Characterization and recognition}\label{ClasseAndRec}

The \textit{subdivision} of a graph $F$, $S(F)$, is the replacement of every edge $uv \in E(F)$ for a new vertex $x_{uv}$ and edges $x_{uv}u$ and $x_{uv}v$.
The \textit{line graph of $F$}, written $L(F)$, is the graph whose vertex set is $E(F)$ and in which two distinct vertices $uv$ and $xy$ are adjacent if and only if they are adjacent in $F$, i.e., $\{u, v\} \cap \{x, y\} \ne \varnothing$.
We say that $L(S(F))$ is a {\em subdivided-line graph}~\cite{Hasunuma}. 
See an example in Figure~\ref{fig:LinhaSubBi}.


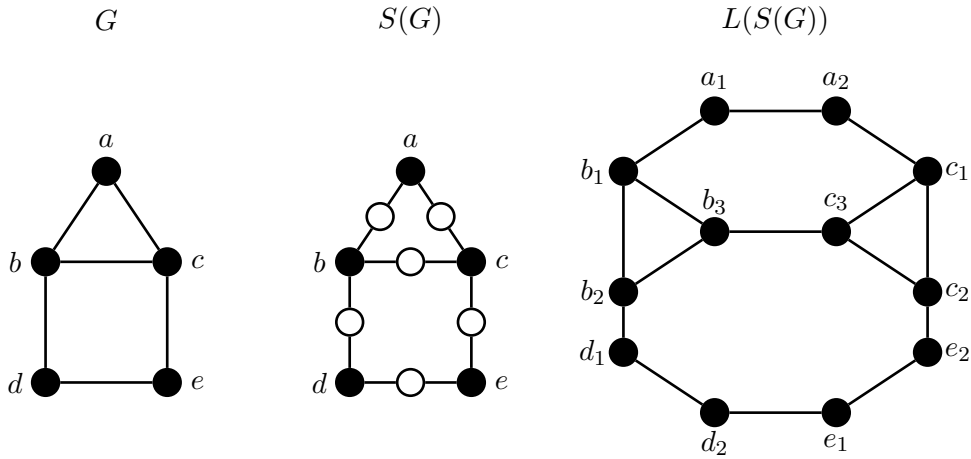
\begin{figure}[!htb]
    
    \vspace{10pt}
    \centering
    
    \begin{tikzpicture}[scale=0.4]
        \pgfsetlinewidth{1pt}
        
        \tikzset{
            vertex/.style={circle,  draw, minimum size=10pt, inner sep=0pt}}

  \begin{scope}

\draw[] (1*2,9) node {$G$};
\draw[] (2,5) node {$a$};
\draw[] (-1, 1) node {$b$};
\draw[] (5, 1) node {$c$};
\draw[] (-1, -3) node {$d$};
\draw[] (5, -3) node {$e$};

            \node [vertex, fill=black] (va) at (2,4){};
            \node [vertex, fill=black] (vb) at (0,1){};
            \node [vertex, fill=black] (vd) at (0,-3){};
            \node [vertex, fill=black] (ve) at (4,-3){};
            \node [vertex, fill=black] (vc) at (4,1){};    
 
            \draw[-] (va) to (vb);
            \draw[-] (vb) to (vd);     
            \draw[-] (vd) to (ve);
            \draw[-] (ve) to (vc);
            \draw[-] (va) to (vc);
            \draw[-] (vb) to (vc);
            
      \end{scope}

\begin{scope}[shift={(10,0)}]

\draw[] (1*2,9) node {$S(G)$};
\draw[] (2,5) node {$a$};
\draw[] (-1, 1) node {$b$};
\draw[] (5, 1) node {$c$};
\draw[] (-1, -3) node {$d$};
\draw[] (5, -3) node {$e$};

            \node [vertex, fill=black] (va) at (2,4){1};
            \node [vertex, fill=black] (vb) at (0,1){2};
            \node [vertex, fill=black] (vd) at (0,-3){3};
            \node [vertex, fill=black] (ve) at (4,-3){4};
            \node [vertex, fill=black] (vc) at (4,1){5}; 

            \node [vertex, fill=white] (vab) at (1,2.5){}; 
            \node [vertex] (vac) at (3,2.5){}; 
            \node [vertex] (vbc) at (2,1){}; 
            \node [vertex] (vbd) at (0,-1){};
            \node [vertex] (vde) at (2,-3){};
            \node [vertex] (vce) at (4,-1){};

            \draw[-] (va) to (vab);
            \draw[-] (vb) to (vab);
            \draw[-] (vb) to (vbc);
            \draw[-] (vc) to (vbc);
            \draw[-] (vc) to (vac);
            \draw[-] (va) to (vac);
            \draw[-] (vb) to (vbd);
            \draw[-] (vbd) to (vd);
            \draw[-] (vd) to (vde);
            \draw[-] (ve) to (vde);
            \draw[-] (ve) to (vce);
            \draw[-] (vc) to (vce);

      \end{scope}

\begin{scope}[shift={(20,0)}]

\draw[] (1*4,9) node {$L(S(G))$};
\draw[] (2,7) node {$a_1$};
\draw[] (6,7) node {$a_2$};
\draw[] (-2,4) node {$b_1$};
\draw[] (-2,0) node {$b_2$};
\draw[] (-2,-2) node {$d_1$};
\draw[] (2,-5) node {$d_2$};
\draw[] (6,-5) node {$e_1$};
\draw[] (10,-2) node {$e_2$};
\draw[] (10,0) node {$c_2$};
\draw[] (10,4) node {$c_1$};
\draw[] (2,3) node {$b_3$};
\draw[] (6,3) node {$c_3$};
                
            \node [vertex, fill=black] (va1) at (2,6){};
            \node [vertex, fill=black] (va2) at (6,6){};
            \node [vertex, fill=black] (vb3) at (2,2){};
            \node [vertex, fill=black] (vc3) at (6,2){};
            \node [vertex, fill=black] (vb1) at (-1,4){};
            \node [vertex, fill=black] (vc1) at (9,4){};
            \node [vertex, fill=black] (vb2) at (-1,0){};
            \node [vertex, fill=black] (vc2) at (9,0){};
            \node [vertex, fill=black] (vd2) at (2,-4){};
            \node [vertex, fill=black] (ve1) at (6,-4){};
            \node [vertex, fill=black] (vd1) at (-1,-2){};
            \node [vertex, fill=black] (ve2) at (9,-2){};

            \draw[-] (va1) to (va2);
            \draw[-] (vc1) to (va2); 
            \draw[-] (vc3) to (vb3);
            \draw[-] (vb3) to (vb1);
            \draw[-] (va1) to (vb1);
            \draw[-] (vc3) to (vc1);
            \draw[-] (vb1) to (vb2);
            \draw[-] (vb3) to (vb2);
            \draw[-] (vc3) to (vc2);
            \draw[-] (vc1) to (vc2);
            \draw[-] (vb2) to (vd1);
            \draw[-] (vd1) to (vd2);
            \draw[-] (ve1) to (ve2);
            \draw[-] (ve2) to (vc2);
            \draw[-] (vd2) to (ve1);

      \end{scope}

\end{tikzpicture}
\label{fig:LinhaSubBi}
\caption{An example of a line graph of subdivision.}
\end{figure}

The Sierpi\'nski graphs were introduced by Klav\v{z}ar and Milutinovi\'{c} as a generalization of the graph of the Tower of Hanoi problem~\cite{Klavzar1997}.
Given integers $p \geq 1$ and $q \geq 1$, the {\em Sierpi\'nski graph} $S(p,q)$ has a vertex for each $p$-tuple that can be formed from $\{1, \ldots, q\}$ and, for two distinct vertices $u = (u_1, u_2,\ldots, u_p)$ and $w = (w_1,w_2,\ldots,w_p)$, $uw\in E(S(p,q))$ if and only if there exists an $h \in \{1,\ldots,p\}$ such that

\begin{enumerate}
	\item $u_t = w_t$ for $t\in \{1,\ldots,h-1\}$;
	
	\item $u_h \neq w_h$;
	
	\item $u_t = w_h$ and $w_t = u_h$ for $t \in \{h + 1, p\}$.
\end{enumerate}

In Figure~\ref{fig:HanoiGrafo}, we show examples of $S(1,3)$, $S(2,3)$, and $S(3,3)$. 


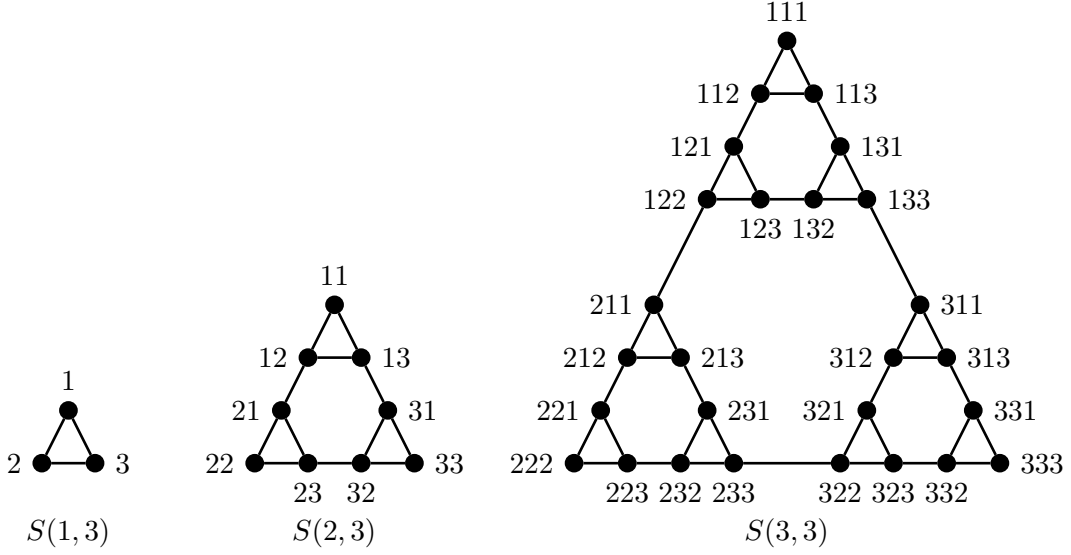
\begin{figure}[!htb]
    
    \vspace{10pt}
    \centering
    
    \begin{tikzpicture}[scale=0.35]
        \pgfsetlinewidth{1pt}
        
        \tikzset{
            vertex/.style={circle,  draw, minimum size=6pt, inner sep=0pt}}

  \begin{scope}[shift={(0,0)}]
           \node (S13) (1) at (1,-4)  [label=above:${S(1,3)}$]{};
                          
            \node [vertex,fill=black] (1) at (1,2)  [label=above:${1}$]{};
            \node [vertex,fill=black] (2) at (0,0)  [label=left:${2}$]{};
            \node [vertex,fill=black] (3) at (2,0)  [label=right:${3}$]{};

            \draw[-] (1) to (2);
            \draw[-] (1) to (3);
            \draw[-] (2) to (3);

      \end{scope}

\begin{scope}[shift={(0,4)}]

  \begin{scope}[shift={(10,0)}]
            \node (S23) (1) at (1,-8)  [label=above:${S(2,3)}$]{};
                          
            \node [vertex,fill=black] (11) at (1,2)  [label=above:${11}$]{};
            \node [vertex,fill=black] (12) at (0,0)  [label=left:${12}$]{};
            \node [vertex,fill=black] (13) at (2,0)  [label=right:${13}$]{};

            \draw[-] (11) to (12);
            \draw[-] (11) to (13);
            \draw[-] (12) to (13);

      \end{scope}

  \begin{scope}[shift={(8,-4)}]

            \node [vertex,fill=black] (21) at (1,2)  [label=left:${21}$]{};
            \node [vertex,fill=black] (22) at (0,0)  [label=left:${22}$]{};
            \node [vertex,fill=black] (23) at (2,0)  [label=below:${23}$]{};

            \draw[-] (21) to (22);
            \draw[-] (21) to (23);
            \draw[-] (22) to (23);

      \end{scope}

  \begin{scope}[shift={(12,-4)}]

            \node [vertex,fill=black] (31) at (1,2)  [label=right:${31}$]{};
            \node [vertex,fill=black] (32) at (0,0)  [label=below:${32}$]{};
            \node [vertex,fill=black] (33) at (2,0)  [label=right:${33}$]{};

            \draw[-] (31) to (32);
            \draw[-] (31) to (33);
            \draw[-] (32) to (33);

      \end{scope}

            \draw[-] (12) to (21);
            \draw[-] (23) to (32);
            \draw[-] (13) to (31);

\end{scope}

\begin{scope}[shift={(2,14)}]

    \begin{scope}[shift={(15,0)}]

          \begin{scope}[shift={(10,0)}]
            \node (S33) (1) at (1,-18)  [label=above:${S(3,3)}$]{};
                          
            \node [vertex,fill=black] (111) at (1,2)  [label=above:${111}$]{};
            \node [vertex,fill=black] (112) at (0,0)  [label=left:${112}$]{};
            \node [vertex,fill=black] (113) at (2,0)  [label=right:${113}$]{};

            \draw[-] (111) to (112);
            \draw[-] (111) to (113);
            \draw[-] (112) to (113);

      \end{scope}

  \begin{scope}[shift={(8,-4)}]

            \node [vertex,fill=black] (121) at (1,2)  [label=left:${121}$]{};
            \node [vertex,fill=black] (122) at (0,0)  [label=left:${122}$]{};
            \node [vertex,fill=black] (123) at (2,0)  [label=below:${123}$]{};

            \draw[-] (121) to (122);
            \draw[-] (121) to (123);
            \draw[-] (122) to (123);

      \end{scope}

  \begin{scope}[shift={(12,-4)}]

            \node [vertex,fill=black] (131) at (1,2)  [label=right:${131}$]{};
            \node [vertex,fill=black] (132) at (0,0)  [label=below:${132}$]{};
            \node [vertex,fill=black] (133) at (2,0)  [label=right:${133}$]{};

            \draw[-] (131) to (132);
            \draw[-] (131) to (133);
            \draw[-] (132) to (133);

      \end{scope}

            \draw[-] (112) to (121);
            \draw[-] (123) to (132);
            \draw[-] (113) to (131);

\end{scope}

 \begin{scope}[shift={(10,-10)}]

    \begin{scope}[shift={(10,0)}]

            \node [vertex,fill=black] (211) at (1,2)  [label=left:${211}$]{};
            \node [vertex,fill=black] (212) at (0,0)  [label=left:${212}$]{};
            \node [vertex,fill=black] (213) at (2,0)  [label=right:${213}$]{};

            \draw[-] (211) to (212);
            \draw[-] (211) to (213);
            \draw[-] (212) to (213);

      \end{scope}

  \begin{scope}[shift={(8,-4)}]

            \node [vertex,fill=black] (221) at (1,2)  [label=left:${221}$]{};
            \node [vertex,fill=black] (222) at (0,0)  [label=left:${222}$]{};
            \node [vertex,fill=black] (223) at (2,0)  [label=below:${223}$]{};

            \draw[-] (221) to (222);
            \draw[-] (221) to (223);
            \draw[-] (222) to (223);

      \end{scope}

  \begin{scope}[shift={(12,-4)}]

            \node [vertex,fill=black] (231) at (1,2)  [label=right:${231}$]{};
            \node [vertex,fill=black] (232) at (0,0)  [label=below:${232}$]{};
            \node [vertex,fill=black] (233) at (2,0)  [label=below:${233}$]{};

            \draw[-] (231) to (232);
            \draw[-] (231) to (233);
            \draw[-] (232) to (233);

      \end{scope}

            \draw[-] (212) to (221);
            \draw[-] (223) to (232);
            \draw[-] (213) to (231);

\end{scope}

\begin{scope}[shift={(20,-10)}]

  \begin{scope}[shift={(10,0)}]

            \node [vertex,fill=black] (311) at (1,2)  [label=right:${311}$]{};
            \node [vertex,fill=black] (312) at (0,0)  [label=left:${312}$]{};
            \node [vertex,fill=black] (313) at (2,0)  [label=right:${313}$]{};

            \draw[-] (311) to (312);
            \draw[-] (311) to (313);
            \draw[-] (312) to (313);

      \end{scope}

  \begin{scope}[shift={(8,-4)}]

            \node [vertex,fill=black] (321) at (1,2)  [label=left:${321}$]{};
            \node [vertex,fill=black] (322) at (0,0)  [label=below:${322}$]{};
            \node [vertex,fill=black] (323) at (2,0)  [label=below:${323}$]{};

            \draw[-] (321) to (322);
            \draw[-] (321) to (323);
            \draw[-] (322) to (323);

      \end{scope}

  \begin{scope}[shift={(12,-4)}]

            \node [vertex,fill=black] (331) at (1,2)  [label=right:${331}$]{};
            \node [vertex,fill=black] (332) at (0,0)  [label=below:${332}$]{};
            \node [vertex,fill=black] (333) at (2,0)  [label=right:${333}$]{};

            \draw[-] (331) to (332);
            \draw[-] (331) to (333);
            \draw[-] (332) to (333);

      \end{scope}

            \draw[-] (312) to (321);
            \draw[-] (323) to (332);
            \draw[-] (313) to (331);

\end{scope}

            \draw[-] (122) to (211);
            \draw[-] (133) to (311);
            \draw[-] (233) to (322);
 
\end{scope}
   
\end{tikzpicture}

\caption{Examples of Sierpi\'nski graphs.} \label{fig:HanoiGrafo}
\end{figure}

Observe that a subdivided-line graph has no simplicial vertices, whereas the number of simplicial vertices of a Sierpi\'nski graph is $q$, which means that they are disjoint graph classes. In the next result, we show that both are proper subclasses of the expanded-clique graphs.

\begin{proposition}
If $H$ is a subdivided-line graph or a Sierpi\'nski graph, then $H$ is expanded-clique.
\end{proposition}

\begin{proof}
For both cases, we need to find a root $G$ and its corresponding $f$.

For a subdivided-line graph $H = L(S(G))$, take $G$ as the root and $f(v_i) = d(v_i)$ for every $v_i \in V(G)$.
Note that the vertex set of $H$ is formed by a set $V_i$ with $d(v_i)$ vertices forming a clique for every $v_i \in V(G)$. Furthermore, if $v_iv_j \in E(G)$, then the vertex $v_{ij}$ of $S(G)$, originated by the edge $v_iv_j$, becomes an edge in $H$ joining one vertex of $V_i$ to an exclusive vertex of $V_j$. Since there are no more edges, we conclude that $H$ is the expanded-clique graph of $(G,f)$.

Now, let $H$ be a Sierpi\'nski graph $S(p,q)$. Note that if $q=1$, then $H$ is the trivial graph. Then, we can assume that $q \ge 2$.
For $p = 1$, note that $H$ is the complete graph with $q$ vertices. Then, we can choose $V(G) = \{v_1\}$ and $f(v_1) = q$. 
For $p > 1$, let $G = S(p-1,q)$ and $f(v_i) = q$ for every $v_i \in V(G)$.
The expanded-clique graph of $(G,f)$ has $q |V(G)|$ vertices and this is the number of vertices of $S(p,q)$ because for any $(p-1)$-tuple that can be formed with $q$ elements, we can form $q$ tuples by adding one coordinate. The definition of Sierpi\'nski graphs and the definition of expanded-clique graphs imply that these $q$ vertices form a clique. Denote by $V_i$ such clique associated with $v_i \in V(S(p-1,q))$. We can write $v_i = (u_1, \ldots, u_{p-1})$. Note that $d(v_i) \in \{q-1,q\}$.
We know that every vertex of $S(p,q)$ and every vertex of a $q$-expanded-clique graph has also degree $q-1$ or $q$.
If $d(v_i) = q$, then since $p \ge 2$ and $q \ge 2$, the definition implies that there are two coordinates of $v_i$ that are different. Then, for every $v_j \in N(v_i)$, there is a vertex in $V_i$ and a vertex in $V_j$ that are adjacent in $S(p,q)$.
If $d(v_i) = q-1$, then all coordinates of $v_i$ are equal. Then for every $v_j \in N(v_i)$, there is a vertex in $V_i$ and a vertex in $V_j$ that are adjacent in $S(p,q)$.
Since there are no more edges in $S(p,q)$ and the edges cited above are precisely the edges of the expanded-clique graph of $(G,f)$, we conclude that $H = S(p,q)$ is an expanded-clique graph.
\end{proof}

Before showing that the line graphs of bipartite graphs form a superclass of the expanded-clique graphs, we recall an useful result.

\begin{proposition} \label{LinhaBip} {\em \cite{Harary74}}
	A graph is a line graph of a bipartite graph if and only if it is (claw,~diamond,~odd-hole)-free.
\end{proposition}

\begin{theorem}\label{GrafoLinhaDoPro}
If $H$ is an expanded-clique graph, then $H$ is a line graph of a bipartite graph.
\end{theorem}

\begin{proof}[proof]
Let $H$ be an expanded-clique graph. By Proposition~\ref{LinhaBip}, it suffices to show that $H$ is (claw, diamond, odd-hole)-free. By the definition of expanded-clique graphs, for every $v \in V(H)$, either $N(v)$ is a clique or there is $u \in N(v)$ such that $N(v) \setminus \{u\}$ is a clique and $u$ has no neighbors in $N(v)$. Since a claw and a diamond have vertices not satisfying this property, we conclude that $H$ is (claw, diamond)-free. Now, let $C_q = u_1 \ldots u_q u_1$ be an induced cycle of $H$ where $q \ge 4$, and let $G$ be a root of $H$.
By the definition of expanded-clique graphs, every vertex $u_i \in V(H)$ belongs to exactly one expanded clique $V_j$, for some $v_j \in V(G)$, and has at most one neighbor outside $V_j$.
Since $C_q$ has no chords, for every vertex $u_i$ of $C_q$, one of its neighbors in $C_q$ belongs to the same expanded clique as $u_i$ and the other belongs to another expanded clique, which implies that $q$ is even and that $H$ is odd-hole-free.
\end{proof}

For $k \ge 3$, a sequence of vertices $u_1, \ldots, u_k$ of a graph $F$  is a {\em chain} if $u_iu_{i+1} \in E(F)$ for $i \in \{1, \ldots, k-1\}$, $d(u_i) = 2$ for $i \in \{2, \ldots, k-1\}$ and it is maximal with these properties. We say that a chain is {\em bad} if $k$ is odd, $d(u_1) \ge 3$ and $d(u_k) \ge 3$. If a chain is not bad, then it is {\em good}. We say that a vertex $v \in V(F)$ is {\em $1$-simplicial} if there is $u \in N(v)$ such that $N(v) \setminus \{u\}$ is a clique and $u$ has no neighbors in $N(v)$. In this case we say that $u$ is the {\em outsider} of $v$.

\begin{theorem} \label{the:charac}
A graph $H$ that is not a cycle is expanded-clique if and only if every vertex of $H$ is simplicial or $1$-simplicial and every
chain of $H$ is good.
\end{theorem}

\begin{proof}
We can assume that $H$ is a non-trivial connected graph.
We begin by considering the case where $H$ is a path $P_k$ for $k \ge 2$.
Since, in this case, every vertex of $H$ is simplicial or 1-simplicial and $H$ does not have bad chains, we have to show a pair $(G,f)$ such that $G$ is the root of $H$ under ther $f$-expanded-clique operation. 
For $k$ even, we choose $G = P_{\frac{k}{2}} = v_1 \ldots v_{\frac{k}{2}}$ and set $f(v_i) = 2$ if $i \in \{1,\ldots,\frac{k}{2}\}$; and for $k$ odd, we choose $G = P_{\lceil\frac{k}{2}\rceil} = v_1 \ldots v_{\lceil\frac{k}{2}\rceil}$ and set $f(v_i) = 2$ if
$i \in \{1,\ldots,\lceil\frac{k}{2} \rceil -1\}$ and
$f(v_{\lceil\frac{k}{2} \rceil}) = 1$.
From now on, we can assume that $H$ is neither a cycle nor a path.

For the necessity, consider that $H$ is the expanded-clique of a pair $(G,f)$.
First, suppose by contradiction that $u_1, \ldots, u_k$ is a bad chain of $H$.
Since $d(u_1) \ge 3$ and $d(u_2) = 2$, we conclude that $u_1$ and $u_2$ belong to different expanded cliques of $H$. By symmetry, $u_{k-1}$ and $u_k$ belong to different expanded cliques of $H$.
Furthermore, the expanded clique containing $u_2$ also contains $u_3$ and no more vertices. 
Then, $\{u_i,u_{i+1}\}$ is an expanded clique for every $i < k$ such that $i$ is even.
Therefore, $u_{k-1}$ and $u_k$ belong to the same expanded clique of $H$, which is a contradiction. Therefore, any chain of $H$ is good.

Now, suppose by contradiction that there is $v \in V(H)$ such that $N(v)$ is not simplicial neither $1$-simplicial. Note that $d(v) \ge 3$.
Then, there are $u,w \in N(v)$ such that $uw \not\in E(H)$ and let $x \in N(v) \setminus \{u,w\}$.
If $u,w \in N_H(x)$, then $H$ has a diamond, which is not possible by Proposition~\ref{LinhaBip} and Theorem~\ref{GrafoLinhaDoPro}.
If $ux, xw \not\in E(H)$, then $H$ has a claw, which is also not possible.
Therefore, every vertex of $N(v) \setminus \{u,w\}$ is adjacent to exactly one vertex of $\{u,w\}$.
Hence, the subgraph of $H$ induced by $N(v)$ has exactly two connected components, say $C_1$ and $C_2$, each one being a complete graph.
Since $v$ is not 1-simplicial, we have that $|V(C_i)| \ge 2$ for $i \in \{1,2\}$.
Without loss of generality, we can assume that the expanded clique containing $v$ is $\{v\} \cup V(C_1)$. Now, we reach a contradiction because $v$ has at least 2 neighbors outside 
the expanded clique containing it, which is not possible in an expanded-clique graph.

For the sufficiency, consider that every chain of $H$ is good and for every $v \in V(H)$, either $v$ is a simplicial or a $1$-simplicial vertex. We will construct a graph $G$ and a function $f$ such that $H$ is the expanded-clique graph of $(G,f)$. In order to do this, we define $S(v)$ for every $v \in V(H)$ as follows.

First, consider the vertices $v$ with degree at least $3$ in any order. If $v$ is a simplicial vertex, then set $S(v) = N(v)$; otherwise, let $u \in N(v)$ such that $N(v) \setminus \{u\}$ is a clique and $u$ has no neighbors in $N(v)$, and set $S(v) = N(v) \setminus \{u\}$.
For the vertices $v$ with degree $2$ in any exists, choose anyone having a neighbor $w$ such that $S(w)$ has already been defined. Denote by $u$ the neighbor of $v$ different of $w$. Then, define $S(v) = \{v,u\}$ and $S(u) = \{v,u\}$ and repeat. Finally, if there are pendant vertices $v$ such that $S(v)$ has not been defined yet, define $S(v) = \{v\}$. The assumptions on $H$ guarantee that $S(v)$ will be defined for every vertex $v \in V(H)$.

Now, consider any ordering $u_1, \ldots, u_p$ of $V(H)$. Then, for $i$ from $1$ to $p$, add to $G$ a vertex $v_i$ if $S(u_i) \ne S(u_j)$ for every $j \in [i-1]$ and set $f(v_i) = |S(u_i)|$.
Finally, add to $G$ the edge $v_iv_j$ if there is an edge joining some vertex of $S(u_i)$ to some vertex of $S(u_j)$. Noting that $S(u_i)$ is a clique for every $u_i \in V(H)$ and that every vertex of $S(u_i)$ has at most one neighbor outside $S(u_i)$, we conclude that $H$ is the expanded-clique graph of $(G,f)$.
\end{proof}

It is clear that Theorem~\ref{the:charac} leads to a polynomial-time algorithm for answering whether a graph $H$ that is not a cycle is expanded-clique. We will present in the sequel a linear-time algorithm for this problem. For completeness, our algorithm considers also the case where $H$ is a cycle. For this purpose, we need of the following result.

\begin{proposition} \label{pro:C4}
If $H$ is an expanded-clique graph, then $H$ is $C_4$-free.
\end{proposition}

\begin{proof}
Let $H$ be the expanded-clique graph of a pair $(G,f)$. Suppose by contradiction that $H$ contains an induced $u_1 u_2 u_3 u_4 u_1$. By the definition, the expanded cliques of $H$ form a partition of $V(H)$. Consider the expanded clique $V_i$ containing $u_1$. We also know that $u_1$ has at most one neighbor outside $V_i$. Without loss of generality, we can assume that $u_4 \in V_i$ and that $u_2$ belongs to an expanded clique $V_j$ different of $V_i$. The same reasoning implies that $u_3 \in V_j$. With these facts, we conclude that there are two edges in $G$ joining the vertices $v_i,v_j \in V(G)$ associated with $V_i$ and $V_j$, which contradicts the assumption that $G$ is a simple graph.
\end{proof}

\newcounter{number_of_lines}

\begin{algorithm}[H] \label{alg:expanded-clique}
	\caption{{\sc Is$\_$expanded$\_$clique}}
	
	\KwIn{A graph $H$ of order $n \ge 3$ that is not a cycle (each adjacency lists is ordered by the vertex number).}
	\KwOut{The root $G$ of $H$ if one there exists}

\If{$H$ is a cycle $C_n$}{ \label{lin:cycle_begin}
    \lIf{$n==3$}{\Return $((\{v_1\},\{\}),3)$}
    \lIf{$n==4$ {\em {\bf or}} $n \ge 5$ odd}
        {\Return {\sc no}}
    \lIf{$n \ge 6$ even}
        {\Return $(C_{\frac{n}{2}},2)$} \label{lin:cycle_end}
}

sort the adjacency lists \label{lin:sort}

\For{$u_i \in H.V$ \label{lin:unmarked_begin}}{
    $u_i.marked =$ {\sc false} \label{lin:unmarked}\\
    $u_i.outsider =$ {\sc null} \\
    $u_i.current = u_i.first\_neighbor$\label{lin:unmarked_end}
}
\For{$u_i \in H.V$}{ \label{lin:main}    
    \If{$u_i.marked ==$ {\sc false} \label{lin:marked_false}}{
        $u_i.marked =$ {\sc true} \label{lin:marked1} \\
        \If{$u_i.deg \ge 3$}{
            \lIf{{\em \ref{pro:s1s}}$(H,u_i) ==$ {\sc false} \label{lin:s1s}} 
            {\Return {\sc no}}            
        }
        \Else{
            \lIf{{\em \ref{pro:chain}}$(H,u_i) ==$ {\sc false} \label{lin:chain}} 
            {\Return {\sc no}}
        }
    }
}
\Return $(G,f)$ \label{lin:end}

\setcounter{number_of_lines}{\value{AlgoLine}}	

\end{algorithm}

\begin{procedure}[H]
\caption{IsSimp1Simp($H,u_i$) \label{pro:s1s}}

\addtocounter{algocf}{-1}	
\setcounter{AlgoLine}{\value{number_of_lines}}

\If{$u_i.outsider = $ {\sc null}}{
    Let $w_1, w_2, w_3$ be the first, the second and the third neighbors of $u_i$ \label{lin:outsider_12_begin}

    \If{$w_1w_2 \not\in E(H)$}{
        \lIf{$w_1w_3 \not\in E(H)$}{ $u_i.outsider = w_1$}
        
        \lElse{$u_i.outsider = w_2$ \label{lin:outsider_12_end}}
    }
    \Else{ \label{lin:outsider_3_begin}
        \For{$w \in u_i.Adj$}{
            \If{$w \ne w_1$ {\bf and} $ww_1 \not\in E(H)$ }{
                $u_i.outsider = w$ \\
                $w.outsider = u_i$ \label{lin:outsider_3_end}
            }
        }
    }          
}
\For{$w \in u_i.Adj$}{ \label{lin:clique_begin}
    \If{$w \ne u_i.outsider$}{
        \For{$z \in u_i.Adj$}{
            \If{$z \ne w$ {\bf and} $z \ne u_i.outsider$} {
                \If{$z.current \ne w$ \label{lin:current}}{
                    \lIf{$z.outsider \ne$ {\sc null} \label{lin:second}}{\Return {\sc false}}
                    $z.outsider = z.current$ \label{lin:outsiderz} \\
                    $z.current.outsider = z$ \label{lin:outsiderz'} \\
                }
                $z.current = z.current.next$ \\ \label{lin:clique_end} 
            }
        }
    }
}
\For{$w \in u.Adj$}{ \label{lin:extra_begin}
    \If{$w \ne u.outsider$}{
        $w.marked = $ {\sc true} \label{lin:marked2}\\
        \If{$w.current \ne $ {\sc null}}{
            \If{$w.outsider ==$ {\sc null}}{
                $w.outsider = w.current$\\
                $w.current.outsider = w$\\
                \lIf{$w.current.next \ne $ {\sc null}}{\Return {\sc false}}
            }
            \lElse{\Return {\sc false}} \label{lin:extra_end}
        }
    }
}

\setcounter{number_of_lines}{\value{AlgoLine}}	

\end{procedure}

\begin{procedure}[H]
\caption{IsGoodChain($H,u$) \label{pro:chain}}

\addtocounter{algocf}{-1}	
\setcounter{AlgoLine}{\value{number_of_lines}}

\If{$u.deg == 1$}{
    $W = $ singleton formed by the only neighbor of $u$\\ \label{lin:W1}
    good$\_$chain = {\sc true} \label{lin:degree1}
}
\Else{
    $W = $ set formed by the two neighbors of $u$\\ \label{lin:W2}
    good$\_$chain = {\sc false}
}
$q = 1$ \\
\For{$w \in W$}{ \label{lin:find_chain_begin}
    \While{{\sc true}} {
        $v = $ neighbor of $w$ different of $u$ \\
        \lIf{$v.deg \ge 3$}{ {\bf break}}
        $q = q+1$ \\
        $v.marked = $ {\sc true} \label{lin:marked_chain}\\
        \If{$v.deg == 1$}{
            good$\_$chain = {\sc true} \\
            break
        }
        $u = w$ \\
        $w = v$ \label{lin:find_chain_end}
    }
}
\lIf{good$\_$chain {\em {\bf or}} $q$ is even}{\Return {\sc true}} \label{lin:good}
\Return {\sc false} \label{lin:bad}

\end{procedure}

\begin{theorem}
Algorithm~$\ref{alg:expanded-clique}$ is correct.
\end{theorem}

\begin{proof}
First, consider that the input graph $H$ is a cycle $C_n$ for $n \ge 3$.
It is clear that $C_3$ is the 3-expanded-clique graph of the trivial graph.
By Proposition~\ref{pro:C4}, we know that $C_4$ is not an expanded-clique graph.
Proposition~\ref{LinhaBip} and Theorem~\ref{GrafoLinhaDoPro} imply that $C_n$ is not expanded-clique graph por $n \ge 5$ odd.
For $n \ge 6$ even, it suffices to note that $C_n$ is the expanded-clique graph of the pair $(C_{\frac{n}{2}},2)$. These cases are considered in lines~\ref{lin:cycle_begin} to~\ref{lin:cycle_end}.

Consider now that $H$ is connected and is not a cycle. By the characterization given in Theorem~\ref{the:charac}, we have to show that the algorithm returns a pair $(G,f)$ where $H$ is the $f$-expanded-clique of $G$ if and only if
every vertex of $H$ is simplicial or 1-simplicial and every chain of $H$ is good. In the for loop beginning in line~\ref{lin:unmarked_begin}, the algorithm set initial values for the variables associated with every vertex $u_i$ of $H$. and $current$. When the $marked$ variable becomes {\sc true}, the algorithm already knows to which expanded clique $u_i$ belongs to, which prevents that an expanded clique be discovered more than once.
The $outsider$ variable will register the outsider of $u_i$ if the algorithm reaches out that $u_i$ is a $1$-simplicial vertex with degree at least $3$.
The auxiliary variable $current$ keeps the current neighbor of $u_i$ during a search in its neighborhood.
Next, in the for loop beginning in line~\ref{lin:main}, the algorithm passes through every vertex of $H$ testing whether it is simplicial or 1-simplicial (line~\ref{lin:s1s}) if its degree is at least $3$. Otherwise, it verifies whether the chain containing it is good (line~\ref{lin:chain}). Therefore, it suffices to show that $(i)$~\ref{pro:s1s}($H,u_i$) returns {\sc false} if and only if $u_i$ is neither simplicial nor 1-simplicial in the case where $u_i.deg \ge 3$, and that $(ii)$~\ref{pro:chain}($H,u_i$) returns {\sc false} if and only if the chain containing $u_i$ is not good in the case where $u_i.deg \le 2$.

\bigskip \noindent $(i)$ From lines~\ref{lin:outsider_12_begin} to~\ref{lin:outsider_3_end}, the algorithm finds out the outsider of vertex $u_i$ if one there exists. Recall that the outsider of $u_i$ is the only neighbor of $u_i$ that is not adjacent to any other neighbor of $u_i$. Then, if the first two neighbors of $u_i$ are not adjacent, then one of them is the outsider of $u_i$. To discover which one is the outsider, it suffices to test whether for any of them if it is adjacent to the third neighbor of $u_i$. These tests are done in lines~\ref{lin:outsider_12_begin} to~\ref{lin:outsider_12_end}. On the other hand, if the first two neighbors of $u_i$ are not adjacent, then none of them is the outsider of $u_i$. Then, it suffices to passes through the adjacency list of $u_i$ testing whether there is some neighbor of $u_i$ non-adjacent to the first neighbor of $u_i$. This verification is done in lines~\ref{lin:outsider_3_begin} to~\ref{lin:outsider_3_end}.

From lines~\ref{lin:clique_begin} to~\ref{lin:clique_end}, the algorithm checks whether any neighbor of $u_i$ different of its outsider is also neighbor of every other neighbor of $u_i$, i.e., if these vertices form a clique. Since the adjacency lists are ordered, we can pass through the adjacency lists of the neighbors of $u_i$ simultaneously.
However, a neighbor $z$ of $u_i$ can be a $1$-simplicial vertex, in which case $z$ has a neighbor $z'$ not in the neighborhood of $u_i$. This possibility is considered in line~\ref{lin:current}. If this occurs, then such $z'$ is the outsider of $z$ and vice-versa and this information is saved in lines~\ref{lin:outsiderz} and~\ref{lin:outsiderz'}.
But if this occurs more than once, then we know that $z$ is neither simplicial nor 1-simplicial, which is tested in line~\ref{lin:second}.

If $u_i$ is a simplicial or a $1$-simplicial vertex, then any neighbor of $u_i$ has $d(u_i)-1, d(u_i)$, or $d(u_i)+1$ neighbors. 
When the algorithm reaches line~\ref{lin:extra_begin}, we have checked that the neighbors of $u_i$ different of its outsider form a clique $C$. Each such vertex can have at most one neighbor outside $C$. We have finished to search the neighborhood of $u_i$, but we have to complete the search at the adjacency list of each neighbor of $u_i$ to guarantee that it does not have more vertices. If it has one neighbor and its outsider is undefined yet, then such vertex is it outsider. If it has 2 or more neighbors, then $H$ is not clique-expanded. These test are done in lines~\ref{lin:extra_begin} to~\ref{lin:extra_end} completing the proof of $(i)$.

\bigskip \noindent $(ii)$ Denote by $u_1, \ldots, u_k$ the chain $C$ containing $u_i$. If $i = 1$ and $d(u_1)= 1$, then $C$ is good by the definition. If this is the case, this information is saved in line~\ref{lin:degree1}.
We use the set $W$ to save the set of neighbors of $u_i$ in line~\ref{lin:W1} or~\ref{lin:W2}. The variable $q$ is used to count the number of internal vertices of the chain.
From lines~\ref{lin:find_chain_begin} to~\ref{lin:find_chain_end}, the algorithm finds out all vertices belonging to such chain. It accomplish this by choosing one neighbor $w$ of $u_i$ and following the path beginning in $u_i$ and containing $w$ until that some vertex with degree different from 2 is found. Then, it repeat the same for the other neighbor of $u_i$ if it exists. 
If the chain is good, the algorithm answers this in line~\ref{lin:good}. Otherwise, it returns bad in line~\ref{lin:bad}.
It is clear that this procedure returns {\sc false} if and only if the chain containing $u_i$ is not good in the case where $u_i.deg \le 2$
\end{proof}

\begin{theorem}
For a graph $H$ of order $n$ and size $m$, Algorithm~$\ref{alg:expanded-clique}$ finishes in $O(n+m)$ steps.
\end{theorem}

\begin{proof}
The case where $H$ is a cycle is considered in lines~\ref{lin:cycle_begin} to~\ref{lin:cycle_end} and it finishes in constant time. From now on we assume that $H$ is a connected graph different of a cycle. Consider that $H$ is represented by adjacency lists and each vertex is associated with an exclusive number from 1 to $n$.

We begin by showing that the sort of the adjacency lists can be done in linear time (line~\ref{lin:sort}). Consider an array $N$ of size $n$. Each position of $N$ is a linked list initially empty. For $i$ from $1$ to $n$, we go through the adjacency list of every vertex $u_i \in V(H)$. For each $w_j \in N(u_i)$, add $u_i$ to the linked list $N[j]$. This is done in $O(n+m)$ steps for all vertices of $H$. Observe that for every $i \in \{1, \dots, n$, $N[i]$ is exaclty the adjacency list of vertex $u_i$ and they appear in ascending order.

Next, it is clear that the for loop of lines~\ref{lin:unmarked_begin} to~\ref{lin:unmarked_end} costs $O(n)$ steps. Since the for loop beginning in line~\ref{lin:main} has $O(n)$ iterations, we have to show that the number of steps of the functions~\ref{pro:s1s} and~\ref{pro:chain} have time complexity $O(n+m)$ over all calls of these procedures.

First, note that the variable $marked$ is set to {\sc false} only once (line~\ref{lin:unmarked}).
Note also that a call to~\ref{pro:s1s}$(H,u_i)$ makes that the variable $marked$ of $u_i$ and of every vertex belonging to the same expanded clique as $u_i$ are set to {\sc true} in lines~\ref{lin:marked1} and~\ref{lin:marked2}, respectively. Therefore, line~\ref{lin:marked_false} guarantees that at most once call to~\ref{pro:s1s} occurs for each expanded clique of $H$. Because of line~\ref{lin:marked_chain}, we can also conclude that at most once call to~\ref{pro:chain} occurs for each chain of $H$. 

Now, note that the cost of a call to~\ref{pro:s1s}$(H,u_i)$ is $O(d(u_i)^2)$. Since the sum of the degrees of the vertices belonging to the expanded clique containing $u_i$ is $O(d(u_i)^2)$, the total complexity of all calls to~\ref{pro:s1s} is $O(m)$.
Analogously, note that the cost of a call to~\ref{pro:chain}$(H,u_i)$ is $O(q)$ where $q$ is the size of the chain containing $q$. Since every internal vertex of a chain has degree 2, the total complexity of all calls to~\ref{pro:chain} is $O(m)$. Therefore, the total time complexity of Algorithm~\ref{alg:expanded-clique} is $O(n+m)$.
\end{proof}

We conclude this section by presenting another characterization of expanded-clique graphs.

\begin{corollary}
A graph is expanded-clique if and only if it is (bad chain,~butterfly,~claw,~$C_4$, diamond,~odd-hole)-free.
\end{corollary}

\begin{proof}
Let $H$ be an expanded-clique graph.
Due to Theorem~\ref{GrafoLinhaDoPro}, we know that $H$ is a line graph of a bipartite graph. Then, Proposition~\ref{LinhaBip} implies that $H$ is (claw, diamond, odd-hole)-free.
By Theorem~\ref{the:charac}, $H$ is bad chain free. Since a butterfly has a vertex that is neither simplicial nor 1-simplicial, Theorem~\ref{the:charac} also implies that $H$ is butterfly free.
Since Proposition~\ref{pro:C4} guarantees that $H$ is $C_4$ free, $H$ is (bad chain,~butterfly,~claw,~$C_4$,~diamond,~odd-hole)-free.
	
Conversely, let $H$ be a graph that is (bad chain,~butterfly,~claw,~$C_4$, diamond,~odd-hole)-free.
If $H$ is a cycle $C_k$, then $k$ is even greater then $4$. It is clear that $C_{2k'}$ for $k' \ge 3$ is an expanded-clique graph. Then, consider that $H$ is not a cycle. Assume by contradiction that $H$ is not an expanded-clique graph.
As $H$ does not contain no bad chains, by Theorem~\ref{the:charac}, $H$ contains some vertex $v$ that is neither simplicial nor 1-simplicial.

If there are $u_1,u_2,u_3$ in $N(v)$ such that $\{u_1,u_2,u_3\}$ is an independent set, then we have a contradiction because $H$ has a claw.
Since $v$ is not simplicial, there are $u_1,u_2 \in N(v)$ such that $u_1u_2 \not\in E(H)$. If some neighbor of $v$ is adjacent to both $u_1,u_2$, we would have a diamond.
Hence, the subgraph of $H$ induced by $N(v)$ has exactly two connected components, say $C_1$ and $C_2$, each one being a complete graph.
Since $v$ is not 1-simplicial, we have that $|V(C_i)| \ge 2$ for $i \in \{1,2\}$. Then, choose $u_1,u_2 \in V(C_1)$
and $u_3,u_4 \in V(C_2)$. These 4 vertices plus $v$ form a butterfly, which is a contradiction.
\end{proof}

\section{The domination problem}\label{SecDomination} 

In this section, we deal with the {\sc Dominating set} problem for $k$-expanded-clique graphs $H$. For $k = 2$, the root $G$ is a path or cycle and can be easily verified that if $|V(G)| \ge 4$, then $\gamma(H) = \lceil\frac{n}{3}\rceil$. For $k=3$, the problem becomes hard as we will see in the sequel.

The {\sc Edge dominating set} problem asks, for a graph $G$ and an integer $\ell$, whether there is a set $E' \subseteq E(G)$ so that $|E'| \leq \ell$ and every edge of $E(G) \setminus E'$ is adjacent to some edge of $E'$.
It is known that the {\sc Edge dominating set} problem is $\NP$-complete for bipartite graphs with maximum degree~$3$~\cite{Yannakakis80}, which means that {\sc Dominating set} is $\NP$-complete for line graphs of bipartite graphs with maximum degree~$4$~(${\cal C}_1$).
It is also known~\cite{ZVERVICH:1995,KRATOCHVIL1994} that the {\sc Dominating set} problem is $\NP$-complete for planar bipartite graphs with maximum degree~$3$ and girth at least $k$ for a fixed $k$~(${\cal C}_2$) and for cubic graphs~(${\cal C}_3$).
Up to our best knowledge, for no proper subclass of these three classes, the {\sc Dominating set} problem is known to be $\NP$-complete.

We show in Theorem~\ref{NPcompletudeDom} that {\sc Dominating set} is \NP-complete for planar bipartite $3$-expanded-clique graphs, which by Theorem~\ref{GrafoLinhaDoPro} is a proper subclass of classes ${\cal C}_1$ and ${\cal C}_2$; and in Theorem~\ref{NPDomLBC} for cubic line graphs of bipartite graphs, a proper subclass of class ${\cal C}_3$.
The proofs of these two results are very similar and since both reduction are done from two variations of the {\sc Dominating set} problem, we begin by presenting a general reduction so that next we complete each proof with the necessary details.

Given a graph $G$, we denote by $\gamma(G)$ the size of a minimum dominating set of $G$. 

\begin{reduction}
Consider an instance $\left\langle G, \ell \right\rangle$ of {\sc Dominating set}. Let $G'$ be the 3-expanded-clique of $G$, and let $H$ be the 3-expanded-clique of $G'$. Set $\ell'= 2|V(G)|+\ell$. Then, $G$ has a dominating set with at most $\ell$ vertices if and only if $H$ has a dominating set with at most $2|V(G)|+\ell$ vertices.
\end{reduction}

\begin{proof}
The pair $\left\langle H, \ell' \right\rangle$ is an instance of {\sc Dominating set} where $H$ is the 3-expanded-clique graph of $G'$. See Figure~\ref{Fig:GrafoReducaoDom} for an example.

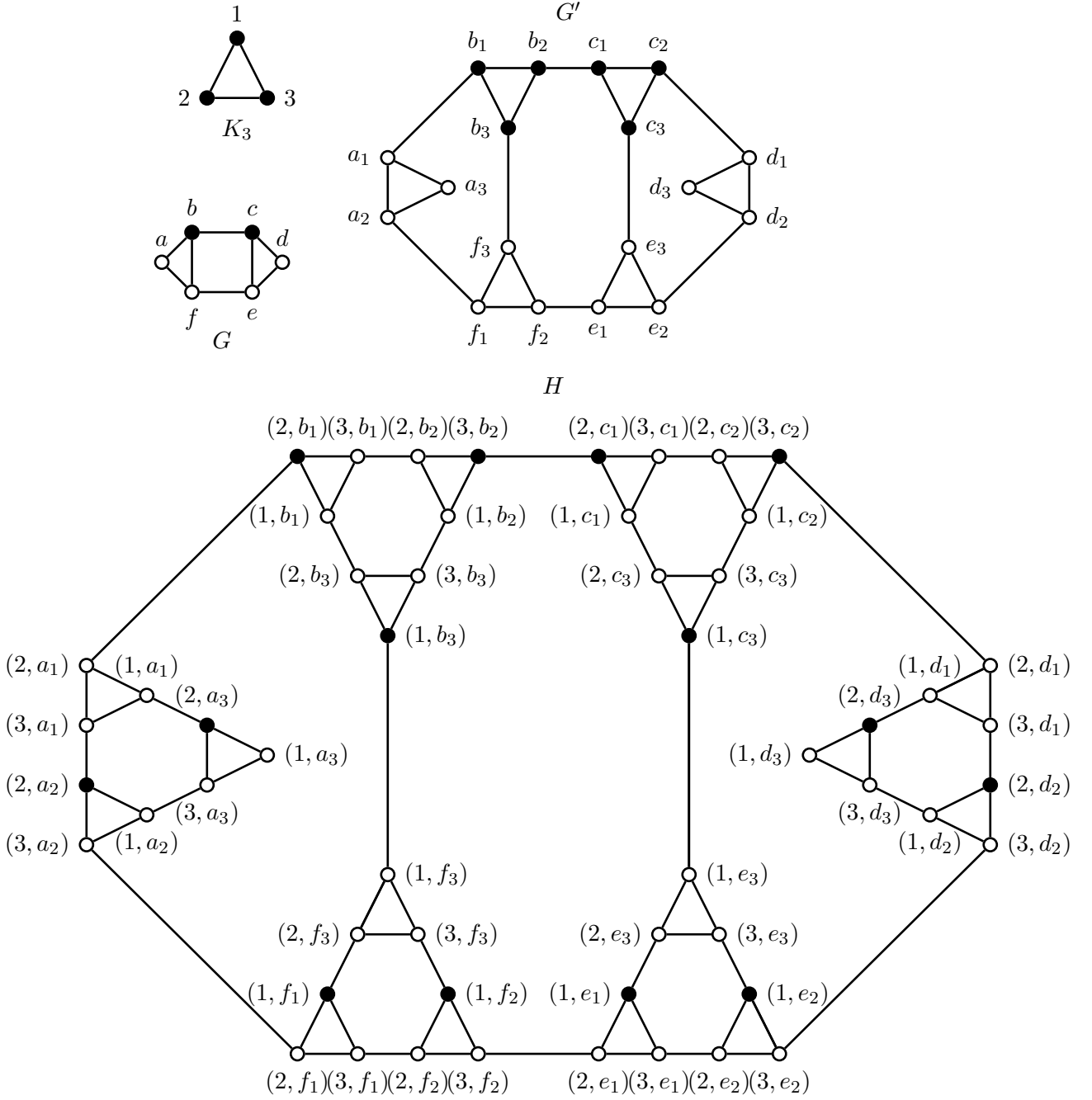
\begin{figure}[!htb]
    
    \vspace{10pt}
    \centering
    
    \begin{tikzpicture}[scale=0.24]
        \pgfsetlinewidth{1pt}
        
        \tikzset{
            vertex/.style={circle,  draw, minimum size=6pt, inner sep=0pt}}
 
 \begin{scope}[shift={(-4,16)}]

            \node [] (G) at (2,-4)  [label=above:${K_3}$]{};           
            \node [vertex, fill=black] (1) at (2,4)  [label=above:${1}$]{};
            \node [vertex, fill=black] (2) at (0,0)  [label=left:${2}$]{};
            \node [vertex, fill=black] (3) at (4,0)  [label=right:${3}$]{};                
 
            \draw[-] (1) to (2);
            \draw[-] (1) to (3);
            \draw[-] (2) to (3);            
      \end{scope}
      
\begin{scope}[shift={(-5,3)}] 

                \node [] (G) at (2,-5)  [label=above:${G}$]{};             
                \node [vertex,fill=black] (b) at (0,4)  [label=above:${b}$]{};
                \node [vertex] (f) at (0,0)  [label=below:${f}$]{};
                \node [vertex] (e) at (4,0)  [label=below:${e}$]{};
                \node [vertex,fill=black] (c) at (4,4)  [label=above:${c}$]{};
                \node [vertex] (a) at (-2,2)  [label=above:${a}$]{};
                \node [vertex] (d) at (6,2)  [label=above:${d}$]{};

                \draw[-] (b) to (c);
                \draw[-] (b) to (f);
                \draw[-] (f) to (e);
                \draw[-] (e) to (c);
                \draw[-] (a) to (b);
                \draw[-] (a) to (f);
                \draw[-] (c) to (d);
                \draw[-] (e) to (d);            
\end{scope}

\begin{scope}[shift={(14,6)}] 

            \node [] (G) at (6,14)  [label=above:${G'}$]{};             
            \node [vertex,fill=black] (b1) at (0,12)  [label=above:${b_1}$]{};
            \node [vertex,fill=black] (b2) at (4,12)  [label=above:${b_2}$]{};
            \node [vertex,fill=black] (b3) at (2,8)  [label=left:${b_3}$]{};
            \node [vertex] (f1) at (0,-4)  [label=below:${f_1}$]{};
            \node [vertex] (f2) at (4,-4)  [label=below:${f_2}$]{};
            \node [vertex] (f3) at (2,0)  [label=left:${f_3}$]{};
            \node [vertex] (e1) at (8,-4)  [label=below:${e_1}$]{};
            \node [vertex] (e2) at (12,-4)  [label=below:${e_2}$]{};
            \node [vertex] (e3) at (10,0)  [label=right:${e_3}$]{};
            \node [vertex,fill=black] (c1) at (8,12)  [label=above:${c_1}$]{};
            \node [vertex,fill=black] (c2) at (12,12)  [label=above:${c_2}$]{};
            \node [vertex,fill=black] (c3) at (10,8)  [label=right:${c_3}$]{};
            \node [vertex] (a1) at (-6,6)  [label=left:${a_1}$]{};
            \node [vertex] (a2) at (-6,2)  [label=left:${a_2}$]{};
            \node [vertex] (a3) at (-2,4)  [label=right:${a_3}$]{};
            \node [vertex] (d1) at (18,6)  [label=right:${d_1}$]{};
            \node [vertex] (d2) at (18,2)  [label=right:${d_2}$]{};
            \node [vertex] (d3) at (14,4)  [label=left:${d_3}$]{};

                \draw[-] (a1) to (a2);
                \draw[-] (b1) to (b2);
                \draw[-] (c1) to (c2);
                \draw[-] (d1) to (d2);
                \draw[-] (e1) to (e2);
                \draw[-] (f1) to (f2);
                \draw[-] (a1) to (b1);
                \draw[-] (c1) to (b2);
                \draw[-] (c2) to (d1);
                \draw[-] (e2) to (d2);
                \draw[-] (e1) to (f2);
                \draw[-] (f1) to (a2);
                \draw[-] (a1) to (a3);
                \draw[-] (a2) to (a3);
                \draw[-] (b2) to (b3);
                \draw[-] (b1) to (b3);
                \draw[-] (c1) to (c3);
                \draw[-] (c2) to (c3);
                \draw[-] (f3) to (f1);
                \draw[-] (f2) to (f3);
                \draw[-] (b3) to (f3);
                \draw[-] (d3) to (d1);
                \draw[-] (d3) to (d2);
                \draw[-] (e3) to (e1);
                \draw[-] (e3) to (e2);
                \draw[-] (c3) to (e3);

      \end{scope}

\begin{scope}[shift={(10,-16)}] 

            \node [vertex] (2a1) at (-22,-6)  [label=left:${(2, a_1)}$]{}; 
            \node [vertex] (3a1) at (-22,-10)  [label=left:${(3, a_1)}$]{}; 
            \node [vertex, fill=black] (2a2) at (-22,-14)  [label=left:${(2, a_2)}$]{};  
            \node [vertex] (3a2) at (-22,-18)  [label=left:${(3, a_2)}$]{};
            \node [vertex] (1a1) at (-18,-8)  [label=above:${(1, a_1)}$]{};
            \node [vertex] (1a2) at (-18,-16)  [label=below:${(1, a_2)}$]{};
            \node [vertex, fill=black] (2a3) at (-14,-10)  [label=above:${(2, a_3)}$]{};
            \node [vertex] (3a3) at (-14,-14)  [label=below:${(3, a_3)}$]{};
            \node [vertex] (1a3) at (-10,-12)  [label=right:${(1, a_3)}$]{};

                \draw[-] (2a1) to (1a1);
                \draw[-] (3a1) to (1a1);
                \draw[-] (2a1) to (3a1);
                \draw[-] (3a1) to (2a2);
                \draw[-] (2a2) to (3a2);
                \draw[-] (2a2) to (1a2);
                \draw[-] (3a2) to (1a2);
                \draw[-] (1a1) to (2a3);
                \draw[-] (1a2) to (3a3);
                \draw[-] (2a3) to (3a3);
                \draw[-] (2a3) to (1a3);
                \draw[-] (3a3) to (1a3);
\end{scope}

\begin{scope}[shift={(26,-40)}, rotate=180] 

            \node [vertex] (3d2) at (-22,-6)  [label=right:${(3, d_2)}$]{}; 
            \node [vertex, fill=black] (2d2) at (-22,-10)  [label=right:${(2, d_2)}$]{}; 
            \node [vertex] (3d1) at (-22,-14)  [label=right:${(3, d_1)}$]{};  
            \node [vertex] (2d1) at (-22,-18)  [label=right:${(2, d_1)}$]{};
            \node [vertex] (1d2) at (-18,-8)  [label=below:${(1, d_2)}$]{};
            \node [vertex] (1d1) at (-18,-16)  [label=above:${(1, d_1)}$]{};
            \node [vertex] (3d3) at (-14,-10)  [label=below:${(3, d_3)}$]{};
            \node [vertex, fill=black] (2d3) at (-14,-14)  [label=above:${(2, d_3)}$]{};
            \node [vertex] (1d3) at (-10,-12)  [label=left:${(1, d_3)}$]{};

                \draw[-] (3d2) to (1d2);
                \draw[-] (2d2) to (1d2);
                \draw[-] (3d2) to (2d2);
                \draw[-] (2d2) to (3d1);
                \draw[-] (3d1) to (2d1);
                \draw[-] (3d1) to (1d1);
                \draw[-] (2d1) to (1d1);
                \draw[-] (2d3) to (1d1);
                \draw[-] (3d3) to (1d2);
                \draw[-] (2d3) to (1d3);
                \draw[-] (2d1) to (1d1);
                \draw[-] (3d3) to (1d3);
                \draw[-] (2d3) to (3d3);

\end{scope}

\begin{scope}[shift={(16,-26)},rotate=90] 

            \node [vertex] (2e1) at (-22,-6)  [label=below:${(2, e_1)}$]{}; 
            \node [vertex] (3e1) at (-22,-10)  [label=below:${(3, e_1)}$]{}; 
            \node [vertex] (2e2) at (-22,-14)  [label=below:${(2, e_2)}$]{};  
            \node [vertex] (3e2) at (-22,-18)  [label=below:${(3, e_2)}$]{};
            \node [vertex, fill=black] (1e1) at (-18,-8)  [label=left:${(1, e_1)}$]{};
            \node [vertex, fill=black] (1e2) at (-18,-16)  [label=right:${(1, e_2)}$]{};
            \node [vertex] (2e3) at (-14,-10)  [label=left:${(2, e_3)}$]{};
            \node [vertex] (3e3) at (-14,-14)  [label=right:${(3, e_3)}$]{};
            \node [vertex] (1e3) at (-10,-12)  [label=right:${(1, e_3)}$]{};

                \draw[-] (1e3) to (2e3);
                \draw[-] (1e3) to (3e3);
                \draw[-] (2e3) to (1e1);
                \draw[-] (1e1) to (2e1);
                \draw[-] (2e1) to (3e1);
                \draw[-] (3e1) to (2e2);
                \draw[-] (3e2) to (2e2);
                \draw[-] (3e2) to (1e2);
                \draw[-] (3e2) to (1e2);
                \draw[-] (1e2) to (3e3);
                \draw[-] (2e3) to (3e3);
                \draw[-] (1e1) to (3e1);
                \draw[-] (1e2) to (2e2);
\end{scope}

\begin{scope}[shift={(-4,-26)},rotate=90] 

            \node [vertex] (2f1) at (-22,-6)  [label=below:${(2, f_1)}$]{}; 
            \node [vertex] (3f1) at (-22,-10)  [label=below:${(3, f_1)}$]{}; 
            \node [vertex] (2f2) at (-22,-14)  [label=below:${(2, f_2)}$]{};  
            \node [vertex] (3f2) at (-22,-18)  [label=below:${(3, f_2)}$]{};
            \node [vertex, fill=black] (1f1) at (-18,-8)  [label=left:${(1, f_1)}$]{};
            \node [vertex, fill=black] (1f2) at (-18,-16)  [label=right:${(1, f_2)}$]{};
            \node [vertex] (2f3) at (-14,-10)  [label=left:${(2, f_3)}$]{};
            \node [vertex] (3f3) at (-14,-14)  [label=right:${(3, f_3)}$]{};
            \node [vertex] (1f3) at (-10,-12)  [label=right:${(1, f_3)}$]{};

                \draw[-] (1f3) to (2f3);
                \draw[-] (1f3) to (3f3);
                \draw[-] (1f1) to (2f3);
                \draw[-] (1f1) to (2f1);
                \draw[-] (2f1) to (3f1);
                \draw[-] (3f1) to (2f2);
                \draw[-] (2f2) to (3f2);
                \draw[-] (1f3) to (2f3);
                \draw[-] (3f2) to (1f2);
                \draw[-] (1f2) to (3f3);
                \draw[-] (2f3) to (3f3);
                \draw[-] (1f1) to (3f1);
                \draw[-] (1f2) to (2f2);

\end{scope}

\begin{scope}[shift={(20,-30)},rotate=-90] 

            \node [vertex, fill=black] (3b2) at (-22,-6)  [label=above:${(3, b_2)}$]{}; 
            \node [vertex] (2b2) at (-22,-10)  [label=above:${(2, b_2)}$]{}; 
            \node [vertex] (3b1) at (-22,-14)  [label=above:${(3, b_1)}$]{};  
            \node [vertex, fill=black] (2b1) at (-22,-18)  [label=above:${(2, b_1)}$]{};
            \node [vertex] (1b2) at (-18,-8)  [label=right:${(1, b_2)}$]{};
            \node [vertex] (1b1) at (-18,-16)  [label=left:${(1, b_1)}$]{};
            \node [vertex] (3b3) at (-14,-10)  [label=right:${(3, b_3)}$]{};
            \node [vertex] (2b3) at (-14,-14)  [label=left:${(2, b_3)}$]{};
            \node [vertex, fill=black] (1b3) at (-10,-12)  [label=right:${(1, b_3)}$]{};

                \draw[-] (2b1) to (3b1);
                \draw[-] (2b2) to (3b1);
                \draw[-] (2b2) to (3b2);
                \draw[-] (2b1) to (1b1);
                \draw[-] (1b1) to (2b3);
                \draw[-] (2b3) to (1b3);
                \draw[-] (1b3) to (3b3);
                \draw[-] (3b3) to (1b2);
                \draw[-] (1b2) to (3b2);
                \draw[-] (2b3) to (3b3);
                \draw[-] (1b1) to (3b1);
                \draw[-] (2b2) to (1b2);
                
                \draw[-] (1b3) to (1f3);

\end{scope}

\begin{scope}[shift={(40,-30)},rotate=-90] 

            \node [vertex, fill=black] (3c2) at (-22,-6)  [label=above:${(3, c_2)}$]{}; 
            \node [vertex] (2c2) at (-22,-10)  [label=above:${(2, c_2)}$]{}; 
            \node [vertex] (3c1) at (-22,-14)  [label=above:${(3, c_1)}$]{};  
            \node [vertex, fill=black] (2c1) at (-22,-18)  [label=above:${(2, c_1)}$]{};
            \node [vertex] (1c2) at (-18,-8)  [label=right:${(1, c_2)}$]{};
            \node [vertex] (1c1) at (-18,-16)  [label=left:${(1, c_1)}$]{};
            \node [vertex] (3c3) at (-14,-10)  [label=right:${(3, c_3)}$]{};
            \node [vertex] (2c3) at (-14,-14)  [label=left:${(2, c_3)}$]{};
            \node [vertex, fill=black] (1c3) at (-10,-12)  [label=right:${(1, c_3)}$]{};

                \draw[-] (2c1) to (3c1);
                \draw[-] (3c1) to (2c2);
                \draw[-] (2c2) to (3c2);
                \draw[-] (2c1) to (1c1);
                \draw[-] (1c1) to (2c3);
                \draw[-] (2c3) to (1c3);
                \draw[-] (1c3) to (3c3);
                \draw[-] (3c3) to (1c2);
                \draw[-] (1c2) to (3c2);
                \draw[-] (2c3) to (3c3);
                \draw[-] (1c1) to (3c1);
                \draw[-] (2c2) to (1c2);
                
                \draw[-] (1c3) to (1e3);

\end{scope}

                \draw[-] (2b1) to (2a1);
                \draw[-] (1c3) to (1e3);
                \draw[-] (3a2) to (2f1);
                \draw[-] (3f2) to (2e1);
                \draw[-] (3e2) to (3d2);
                \draw[-] (2d1) to (3c2);
                \draw[-] (3b2) to (2c1);
                \node [] (G) at (19,-5)  [label=above:${H}$]{};

\end{tikzpicture}

\caption{Graph resulting from polynomial transformation.}
\label{Fig:GrafoReducaoDom}
\end{figure}

For $u \in V(G)$, denote the expanded clique of $G'$ associated with $u$ by $\{u_1,u_2,u_3\}$, and for $i \in [3]$ and $u_i \in V(G')$, denote the expanded clique of $H$ associated with $u_i$ by $\{u_{i,1},u_{i,2},u_{i,3}\}$. For $u \in V(G)$, denote by $H_u$ the subgraph of $H$ induced by $\{u_{i,j} : i \in [3]$ and $j \in [3]\}$. See Figure~\ref{Fig:Gadget}.

\begin{figure}[!htb]
    
    \vspace{10pt}
    \centering
    
    \begin{tikzpicture}[scale=0.3]
        \pgfsetlinewidth{1pt}
        
        \tikzset{
            vertex/.style={circle,  draw, minimum size=10pt, inner sep=0pt}}
 
\begin{scope}[shift={(-54,-26)},rotate=90] 

            \node [] (G) at (-25,-12)  [label=above:$(i)$]{};
            \node [] (G) at (-25,-30)  [label=above:$(ii)$]{};
            \node [] (G) at (-25,-48)  [label=above:$(iii)$]{};

            \node [vertex, fill=gray] (2e1) at (-22,-6) {}; 
            \node [vertex, fill=gray] (3e1) at (-22,-10) {}; 
            \node [vertex, fill=gray] (2e2) at (-22,-14)  {};  
            \node [vertex, fill=gray] (3e2) at (-22,-18)  {};
            \node [vertex, fill=black] (1e1) at (-18,-8)  {};
            \node [vertex, fill=black] (1e2) at (-18,-16)  {};
            \node [vertex, fill=gray] (2e3) at (-14,-10)  {};
            \node [vertex, fill=gray] (3e3) at (-14,-14)  {};
            \node [vertex] (1e3) at (-10,-12)  {};
            
                \draw[-] (1e3) to (2e3);
                \draw[-] (1e3) to (3e3);
                \draw[-] (2e3) to (1e1);
                \draw[-] (1e1) to (2e1);
                \draw[-] (2e1) to (3e1);
                \draw[-] (3e1) to (2e2);
                \draw[-] (3e2) to (2e2);
                \draw[-] (3e2) to (1e2);
                \draw[-] (3e2) to (1e2);
                \draw[-] (1e2) to (3e3);
                \draw[-] (2e3) to (3e3);
                \draw[-] (1e1) to (3e1);
                \draw[-] (1e2) to (2e2);
\end{scope}

\begin{scope}[shift={(-36,-26)},rotate=90] 

            \node [vertex, fill=black] (2f1) at (-22,-6)  {}; 
            \node [vertex, fill=gray] (3f1) at (-22,-10)  {}; 
            \node [vertex, fill=gray] (2f2) at (-22,-14)  {};  
            \node [vertex, fill=black] (3f2) at (-22,-18)  {};
            \node [vertex, fill=gray] (1f1) at (-18,-8)  {};
            \node [vertex, fill=gray] (1f2) at (-18,-16)  {};
            \node [vertex, fill=gray] (2f3) at (-14,-10)  {};
            \node [vertex, fill=gray] (3f3) at (-14,-14) {};
            \node [vertex, fill=black] (1f3) at (-10,-12)  {};

                \draw[-] (1f3) to (2f3);
                \draw[-] (1f3) to (3f3);
                \draw[-] (1f1) to (2f3);
                \draw[-] (1f1) to (2f1);
                \draw[-] (2f1) to (3f1);
                \draw[-] (3f1) to (2f2);
                \draw[-] (2f2) to (3f2);
                \draw[-] (1f3) to (2f3);
                \draw[-] (3f2) to (1f2);
                \draw[-] (1f2) to (3f3);
                \draw[-] (2f3) to (3f3);
                \draw[-] (1f1) to (3f1);
                \draw[-] (1f2) to (2f2);

\end{scope}

\begin{scope}[shift={(-18,-26)},rotate=90] 

            \node [vertex] (2f1) at (-22,-6)  {}; 
            \node [vertex] (3f1) at (-22,-10)  {}; 
            \node [vertex] (2f2) at (-22,-14) {};  
            \node [vertex] (3f2) at (-22,-18)  {};
            \node [vertex] (1f1) at (-18,-8) [label=left:$x$] {};
            \node [vertex] (1f2) at (-18,-16) [label=right:$y$] {};
            \node [vertex] (2f3) at (-14,-10)  {};
            \node [vertex] (3f3) at (-14,-14)  {};
            \node [vertex] (1f3) at (-10,-12) {};

                \draw[-] (1f3) to (2f3);
                \draw[-] (1f3) to (3f3);
                \draw[-] (1f1) to (2f3);
                \draw[-] (1f1) to (2f1);
                \draw[-] (2f1) to (3f1);
                \draw[-] (3f1) to (2f2);
                \draw[-] (2f2) to (3f2);
                \draw[-] (1f3) to (2f3);
                \draw[-] (3f2) to (1f2);
                \draw[-] (1f2) to (3f3);
                \draw[-] (2f3) to (3f3);
                \draw[-] (1f1) to (3f1);
                \draw[-] (1f2) to (2f2);

\end{scope}

\end{tikzpicture}

\caption{Vertex domination in $H_u$.}
\label{Fig:Gadget}
\end{figure}
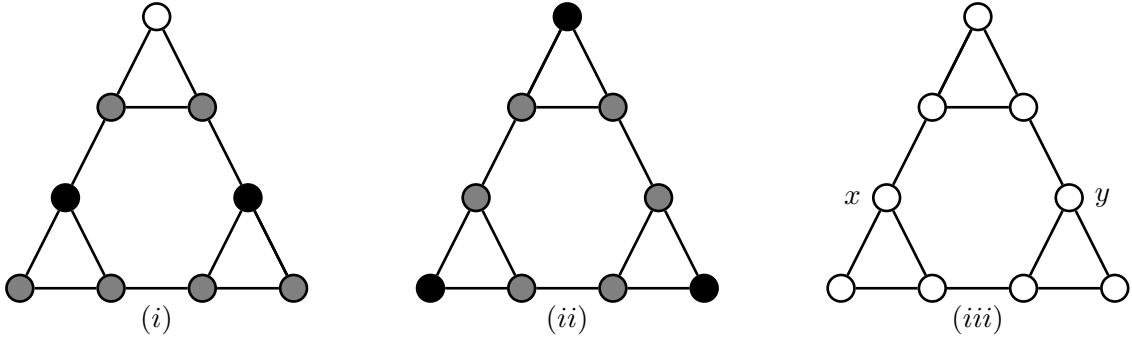

\begin{claim} \label{cla:1}
	For every $u \in V(G)$, it holds that $\gamma(H_u) = 3$.
\end{claim}
	
\begin{proofClaim}{cla:1}
	We know that $\gamma(H_u) \ge 3$ because $\Delta(H_u) = 3$ and $|V(H_u)| = 9$.
	On the other hand, note that any set formed by one vertex of each $K_3$ is a dominating set, and, therefore,  $\gamma(H_u)=3$. See Figure~\ref{Fig:Gadget}-$(ii)$.
\end{proofClaim}

\begin{claim} \label{cla:2}
If $D$ is a dominating set of $H$, then $|V(H_u) \cap D| \ge 2$ for every $u \in V(G)$. Furthermore, if $|V(H_u)\cap D| = 2$ for some $u \in V(G)$, then $V(H_u) \setminus N[V(H_u) \cap D]$ has only one vertex and such vertex is dominated by some vertex in $H_v$ whose $|V(H_v)\cap D| \ge 3$.
\end{claim}
	
\begin{proofClaim}{cla:2}
For any $u \in V(G)$, denote by $U_3$ the subset of $V(H_u)$ having neighbors only in $H_u$ and write $U_2 = V(H_u) \setminus U_3$.
Since $\Delta(H_u) = 3$ and $|U_3| = 6$, we conclude that $|V(H_u) \cap D| \ge 2$.
Now, consider that $V(H_w) \cap D = \{x,y\}$ for some $w \in V(G)$. Observe that if $x \in W_2$, then $D$ would not be a dominating set of $H$. Therefore, by symmetry, we can assume that $x,y$ are the vertices depicted in Figure~\ref{Fig:Gadget}-$(iii)$
and $V(H_w) \setminus N[V(H_w) \cap D]$ has only one vertex and such vertex is dominated by some vertex in $H_v$ for $v \ne w$.
Since a vertex of $V_2$ belongs to $D$, we have that $|V(H_v)\cap D| \ge 3$.
\end{proofClaim}
	
Now, we shall prove that $G$ has a dominating set with at most $\ell$ vertices if and only if $H$ has a dominating set with at most $2|V(G)|+\ell$ vertices.
	
\bigskip
\noindent$(\Rightarrow)$ Consider that $D$ is a dominating set of $G$ with $|D| \leq \ell$.
Starting with $D'$ empty, for each $v\in D$, add to $D'$ the vertices of $H_v$ with degree 2.
Now, for each $u\in V (G)\setminus D$, let $v \in D$ such that $uv \in E(G)$. Denote by $x$ the vertex of $H_u$ having a neighbor in $H_v$. Then, at to $D'$ the two vertices of degree 3 in $H_u$ having a common neighbor with $x$.
Observe that $|D'|=2(|V(G)|-|D|)+3|D|=2|V(G)|+|D|\leq 2|V(G)|+\ell$ and, furthermore, $D'$ is a dominating set of $H$.
	
\bigskip
\noindent$(\Leftarrow)$ Now, let $D'$ be a dominating set of $H$ such that $|D'|\leq 2|V(G)|+\ell$.
By Claim~\ref{cla:2}, we know that $|V(H_u)\cap D'| \ge 2$ for every $u \in V(G)$. Therefore, the number of vertices $v \in V(G)$ such that $|V(H_v)\cap D'| \ge 3$ is at most $\ell$.
Claim~\ref{cla:2} also says that if $|V(H_u)\cap D'| = 2$ for $u \in V(G)$, then $H_u$ has a vertex dominated by some vertex in $H_v$ where $|V(H_v)\cap D_{H}| \ge 3$, which means that $uv$ is an edge of $G$.
Therefore, choosing set $D \subseteq V(G)$ composed by vertices $v$ such that $|V(H_v)\cap D_{H}| \ge 3$ we have a dominating set $D$ of $G$ with size at most $\ell$.
\end{proof}

\begin{theorem}\label{NPcompletudeDom}
{\sc Dominating set} is $\NP$-complete for planar bipartite $3$-expanded-clique graphs with girth at least $k$ for a fixed $k$.
\end{theorem}

\begin{proof}[proof]
Since {\sc Dominating set} belongs to $\NP$ for general graphs, that condition holds for our particular case. For the hardness part, we consider the {\sc Dominating set} problem restricted to planar bipartite graphs with maximum degree $3$ and girth at least $k$ for a fixed $k$ since this version is $\NP$-complete~\cite{ZVERVICH:1995}. Let $\left\langle G, \ell \right\rangle$ be an instance of this problem.
Observe that graph $H$ constructed by 
applying Reduction~A to $\left\langle G, \ell \right\rangle$, is a planar bipartite 3-expanded-clique graph, which means that 
{\sc Dominating set} is $\NP$-complete for planar $3$-expanded-clique graphs.
\end{proof}

\begin{theorem}\label{NPDomLBC}
{\sc Dominating set} is $\NP$-complete for cubic line graphs of bipartite graphs. 
\end{theorem}

\begin{proof}[proof]
As note in the previous result, this problem belongs to $\NP$. We consider the {\sc Dominating set} problem restricted to cubic graphs since this version is also $\NP$-complete~\cite{KRATOCHVIL1994}. Let $\left\langle G, \ell \right\rangle$ be an instance of this problem.
The graph $H$ constructed by 
applying Reduction~A to $\left\langle G, \ell \right\rangle$ is a 3-expanded-clique graph. By Theorem~\ref{GrafoLinhaDoPro}, $H$ is a line graph of a bipartite graph.
It is easy to see that $H$ is a cubic graph, which means that 
{\sc Dominating set} is $\NP$-complete for planar $3$-expanded-clique graphs.
\end{proof}

\begin{proposition} \label{pro:smaller2}
Let $H$ be the expanded-clique graph of a graph $G$ and function $f$. If $S$ is a dominating set of $H$, then there is a dominating set $S'$ of $H$ such that $|S'| \le |S|$ and $|S' \cap V_i| \le 1$ for every $v_i \in V(G)$.
\end{proposition}

\begin{proof}
Let $S$ be a dominating set of $H$. Apply the following process. If $|S \cap V_i| \le 1$ for every $v_i \in V(G)$, then set $S' = S$ and stop. Otherwise, there is $v_i \in V(G)$ such that $|S \cap V_i| \ge 2$. Let $v_{i,j} \in S \cap V_i$. If $|S \cap V_j| \ge 1$, then $S \setminus \{v_{i,j}\}$ is also a dominating set of $H$. Then, redefine $S$ as $S \setminus \{v_{i,j}\}$ and repeat. If $|S \cap V_j| = 0$, then $(S \setminus \{v_{i,j}\}) \cup \{v_{j,i}\}$ is also a dominating set of $H$. Then, redefine $S$ as $(S \setminus \{v_{i,j}\}) \cup \{v_{j,i}\}$ and repeat. Observe that this process eventually finishes with a set $S'$ such that $|S'| \le |S|$ and $|S' \cap V_i| \le 1$ for every $v_i \in V(G)$.
\end{proof}

\begin{proposition} \label{PolyDominating}
If $H$ is the expanded-clique graph of a pair $(G,f)$ where $f(v) > d_G(v)$ for every $v \in V(G)$, then $\gamma(H) = |V(G)|$.
\end{proposition}

\begin{proof}
Note that for every $v_i \in V(G)$, $V_i$ has a simplicial vertex, which implies that any dominating set $D$ of $H$ satisfies $D \cap V_i \ne \emptyset$. On the other hand, if we choose a simplicial vertex of each set $V_i$, we form a dominating set of $H$, which means that $\gamma(H) = |V(G)|$.
\end{proof}

A {\em $2$-independent set} in a graph $G$ is a subset $I$ of the vertices such that the distance between any two vertices of $I$ in $G$ is at least three. We denote by $\alpha_2(G,f)$ the maximum cardinality of a set $S$ such that every $v \in S$ satisfies $f(v) = d_G(v)$ and $S$ is a $2$-independent set of $G$.

\begin{theorem} \label{the:dom}
Let $H$ be the expanded-clique graph of a pair $(G,f)$. Then, $\gamma(H) + \alpha_2(G,f) = |V(G)|$.
\end{theorem}

\begin{proof}
Let $S$ be a dominating set of $H$. By Proposition~\ref{pro:smaller2}, there is a dominating set $S'$ of $H$ such that $|S'| \le |S|$ and $|S' \cap V_i| \le 1$ for every $v_i \in V(G)$.
It is clear that $|S'| \le n$. Note that for any $v_i \in V(G)$, if $f(v_i) > d_G(v_i)$, then $|S' \cap V_i| = 1$. Therefore, if $|S' \cap V_i| = 0$, then $f(v_i) = d_G(v_i)$.
Observe that if $f(v_i) = d_G(v_i)$, $S' \cap V_i = \varnothing$ and $v_iv_j \in E(G)$, then vertex $v_{j,i}$ belongs to $S$. Hence, the vertices $v_i$ with $|S' \cap V_i| = 0$ form a $2$-independent set $T$ in $G$ with cardinality $|T| = |V(G)| - |S'|$ containing only vertices $v_i$ such that $f(v_i) = d_G(v_i)$. Since $|S'| \le |S|$, we conclude that $\alpha_2(G,f) \ge |V(G)| - \gamma(H)$.

Now, let $T$ be a $2$-independent set of $G$ containing only vertices $v_i$ such that $f(v_i) = d_G(v_i)$ and set $S = \varnothing$. If $v_i \in T$, then for every $v_j \in N_G(v_i)$, add $v_{j,i}$ to $S$. If $v_j$ is not in $T$ neither has a neighbor in $T$, then choose a vertex of $V_j$ and add it to $S$. Since no vertex of $G$ has more than one neighbor in $T$, we have that $S$ is a dominating set of $H$ with $|V(G)| - |T|$ vertices. Therefore, $\gamma(H) \le |V(G)| - \alpha_2(G,f)$.
\end{proof}

\begin{theorem} \label{cor:lower}
If $H$ is the $\Delta$-expanded-clique graph of $G$, then $\frac{|V(G)|\Delta}{\Delta+1} \le \gamma(H) \le |V(G)|$.
\end{theorem}

\begin{proof}
The upper bound follows from Proposition~\ref{pro:smaller2}.
For the lower bound, let $S$ be a minimum dominating set of $H$.
Using Proposition~\ref{pro:smaller2} again, we can assume that $|S\cap V_j| \leq 1$ for every $v_j \in V(G)$. Denote by $X$ the set formed by the vertices $v_i \in V(G)$ such that $V_i \cap S = \emptyset$. Note that for every $v_i \in X$, we have that every $d(v_i) = \Delta(G)$, that every vertex $v_k \in N(v_i)$ is such that $|V_k \cap S| = 1$, and that $N(v_i) \cap N(v_j) = \emptyset$.  
On the other hand, for every $v_i \in S \setminus N[X]$, it holds that $|V_i \cap S| = 1$. From these facts, we can write

\[\frac{|S|}{|V(G)|} = \frac
{\underset{v_i \in X}{\sum} \Delta + \underset{v_i \in S \setminus N[X]}{\sum} 1}
{\underset{v_i \in X}{\sum} (\Delta + 1 )+ \underset{v_i \in S \setminus N[X]}{\sum} 1} \ge \frac{\Delta}{\Delta+1}
\]

\noindent which means that the lower bound also holds.
\end{proof}

A consequence of Proposition~\ref{PolyDominating} and Corollary~\ref{cor:lower} is that given a $k$-expanded-clique graph $H$, a set containing one vertex of each expanded clique of $H$ is a dominating set which is at most $\frac{1 + \Delta}{\Delta}$ from a minimum.

\end{document}